\documentclass[11pt]{article}
\usepackage[T1]{fontenc}
\usepackage{lmodern}
\usepackage[letterpaper,margin=1in]{geometry}
\usepackage{amsmath,amssymb,amsthm,mathtools}
\usepackage{microtype}
\usepackage{booktabs,array}
\usepackage[hidelinks]{hyperref}
\usepackage{enumitem}
\usepackage{fancyhdr}
\allowdisplaybreaks[1]
\newtheorem{theorem}{Theorem}[section]
\newtheorem{lemma}[theorem]{Lemma}
\newtheorem{proposition}[theorem]{Proposition}
\newtheorem{corollary}[theorem]{Corollary}
\theoremstyle{definition}
\newtheorem{example}[theorem]{Example}
\theoremstyle{remark}

\newcommand{\C}{\mathbb C}
\newcommand{\dc}{\operatorname{dc}}
\newcommand{\bdc}{\overline{\operatorname{dc}}}
\newcommand{\sdc}{\operatorname{sdc}}
\newcommand{\bsdc}{\overline{\operatorname{sdc}}}
\newcommand{\rank}{\operatorname{rank}}
\newcommand{\adj}{\operatorname{adj}}
\newcommand{\tr}{\operatorname{tr}}
\newcommand{\F}{F_n}

\newcommand{\coeff}[2]{\bigl[#1\bigr]\,#2}
\hypersetup{
  pdftitle={A short proof of a quadratic lower bound on the border determinantal complexity of the power sum over C},
  pdfauthor={Karthik Sheshadri},
  pdfsubject={An intuition-first exposition of the power-sum lower bound},
  pdfkeywords={determinantal complexity, border complexity, power sum, Bezout}
}
\title{A short proof of a quadratic lower bound on the border
determinantal complexity of $x_1^n+\cdots+x_n^n$ over $\C$\\[6pt]
\large A simple human exposition}
\author{Karthik Sheshadri\\[3pt]
  \small Independent researcher}
\date{October 2026}
\begin{document}
\maketitle

\begin{abstract}
I give a short proof that the border determinantal complexity of
$F_n=x_1^n+\cdots+x_n^n$ over $\C$ is at least $(n-1)^2/(4e)$,
and at least $(n-1)^2/(2e)$ in the symmetric model. The same argument shows that the border determinantal complexity of
$x_1^d+\cdots+x_N^d$ is $\Theta(Nd)$ for all $N\ge3$ and $d\ge2$.

The earlier article arXiv:2606.13628, by the same author, introduced the first superlinear lower bounds and includes verification by Lean formalization. However, the AI written exposition was hard to follow. Kumar and Volk (arXiv:2609.34462) give a short proof of the exact ordinary case, but do not address the border or symmetric cases.

In this article, I provide a greatly simplified exposition of the border and
symmetric cases, using the implicit function theorem and an explicit
section of simple solutions in place of the conormal specialization
arguments of arXiv:2606.13628. I choose an explicit polynomial system with $n(n-1)^{n-2}$ simple solutions for $F_n$. I then use the implicit function theorem to ensure that each point persists under small perturbations. Finally, kernel equations and the multihomogeneous B\'ezout theorem bound the number of these points for a small determinant.

\textbf{The exposition is due to the human author. The mathematical simplification was obtained in discussion with ChatGPT Astra and refined in conversation with Claude Opus 5.5}.
\end{abstract}

\section{Introduction}
\label{sec:introduction}

Consider the polynomial
\[
  \F(x_1,\ldots,x_n)=x_1^n+\cdots+x_n^n.
\]
which is easy to specify and evaluate. I study the question of whether it can be the determinant of a small matrix whose entries have degree at most one (affine linear entries).

The determinantal representation of a polynomial $f$ has
the form
\begin{equation}
  f(x)=\det A(x),\qquad
  A(x)=A_0+x_1A_1+\cdots+x_nA_n,
  \label{eq:representation}
\end{equation}
where the $A_i$ are constant $m\times m$ matrices over $\C$. The least
possible $m$ is the \emph{determinantal complexity} of $f$, denoted by
$\dc(f)$. Constant terms in the entries are allowed.

An elementary lower bound is immediately apparent; a determinant of size $m$ has degree
at most $m$, so $\dc(\F)\ge n$. There are also representations of size
$O(n^2)$. For example, one can build a small block for each power $x_i^n$ and then
combine the blocks to add their contributions. For reference, I show such a
construction in Section~\ref{sec:upper}.

The challenge involved in proving a lower bound is to show that a more economical representation cannot exist. This is challenging for various reasons: a representation need not assign a separate block to each
variable. Every entry may involve every variable, and the determinant
expansion may contain extensive cancellations. The fact that one
construction uses roughly $n$ entries for each of $n$ powers is
therefore not a lower-bound argument. I need a property of the polynomial
that every representation must respect.

Repeated squaring would compute each $x_i^n$ using
$O(\log n)$ multiplications. Thus $\F$ has a small arithmetic circuit.
The question here is how much space is needed when all of that
computation must fit into one affine-linear determinant.

Determinantal complexity is a central model in algebraic complexity.
Valiant's work~\cite{Val79} placed it in a general theory of algebraic
computation, and the permanent versus determinant problem has been a
main source of lower-bound questions. For an $n\times n$ permanent,
Mignon and Ressayre~\cite{MR04} proved a lower bound of order $n^2$.
That polynomial has $n^2$ variables. For the power sum above, both the
degree and the number of variables are $n$, so a quadratic lower bound
goes beyond both of these basic parameters. The earlier bound of Kumar
and Volk~\cite{KV21} for this family was linear in $n$.

I also allow approximation. Suppose that a sequence of polynomials
$f_j$, each having a size-$m$ determinantal representation, converges
coefficient by coefficient to $f$. The least such $m$ is the
\emph{border determinantal complexity}, written $\bdc(f)$.

This agrees with the Zariski-closure definition because the size-$m$
determinants form a constructible set, whose Euclidean and Zariski
closures over $\C$ coincide~\cite[Theorems~4.12 and~4.15]{BI25}.

Note that approximation can only make
the problem easier, so
\[
  \bdc(f)\le \dc(f).
\]

The matrices representing the $f_j$ need not converge; it's possible that their entries become arbitrarily large while their determinants still converge.

If all the coefficient matrices in~\eqref{eq:representation} are
symmetric, I write $\sdc(f)$ and $\bsdc(f)$ for the corresponding
symmetric quantities.
Symmetric here means transpose-symmetric, not Hermitian.

\begin{theorem}
\label{thm:main}
For every integer $n\ge3$, over $\C$,
\[
  \dc(\F)\ \ge\ \bdc(\F)\ \ge\ \frac{(n-1)^2}{4e},
  \qquad
  \sdc(\F)\ \ge\ \bsdc(\F)\ \ge\ \frac{(n-1)^2}{2e}.
\]
In particular, all four complexities have order $n^2$.
\end{theorem}

This article gives a direct exposition of the lower bounds in the
longer paper~\cite{She26}. Kumar and Volk~\cite{KV26} give a short
proof of the quadratic lower bound in the exact ordinary model.

The main tools I rely on are linear algebra, the implicit
function theorem, and a form of B\'ezout's theorem that keeps track of
different groups of variables. I first find the points to count,
then show why every nearby polynomial must retain at least that many.

The count in my original proof~\cite{She26} came from the degree of
the dual variety of the Fermat hypersurface.
Kumar and Volk's slice~\cite{KV26} made the points explicit.
Here I choose the slice so that every point is simple; this is what
lets the points persist under perturbation and carries the argument
to the border.

\section{Overview of the proof}
\label{sec:idea}

To prove a lower bound, I would like to count something that is large
for $\F$ but limited for a small determinant. The polynomial itself
does not suggest such a count: it has only $n$ monomials. Its
derivatives are more useful:
\[
  \frac{\partial \F}{\partial x_i}=n x_i^{n-1}.
\]
An equation such as $x_i^{n-1}=1$ has $n-1$ complex solutions.
If I can make many coordinates satisfy independent equations of this
kind, I obtain roughly $n^n$ points. Even after taking an $n$th root,
that is a substantial count.

This idea is well-known in degree-based circuit
lower bounds; see, for example, Baur and Strassen~\cite{BS83}.
In that approach, one expresses a computation by polynomial equations,
introduces a variable for each intermediate value, and then uses the small
degrees of those equations to bound the number of solutions. Note here that a small
computation cannot produce arbitrarily many isolated solutions.
To apply this to a determinant, I need a description that uses its particular
structure.

\subsection{Derivatives on the zero set}

Differentiating $f=\det A$ brings in
the cofactors of $A$, which can still have degree $m-1$.
Away from $f=0$, they are encoded by $\adj A=fA^{-1}$, which is a full
$m\times m$ matrix. The hope is to find points where that matrix
collapses to something smaller.

However, at a smooth zero $a$ of $f$, the situation is better. The matrix $A(a)$
has rank exactly $m-1$. Its left kernel and right kernel are therefore
lines. Choose nonzero vectors $u$ and $v$ on those lines, so
\[
  u^{\mathsf T}A(a)=0,\qquad A(a)v=0.
\]

Observe that the adjugate of $A(a)$ has rank one and factors as
$\lambda vu^{\mathsf T}$ for some nonzero scalar $\lambda$. Jacobi's
formula then gives
\begin{equation}
  \partial_i f(a)=\lambda u^{\mathsf T}A_i v
  \qquad (1\le i\le n).
  \label{eq:keyidentity}
\end{equation}

For a proof of these facts, refer to Section~\ref{sec:kernels}.

The cofactors are now controlled
by two vectors of length $m$. Smoothness matters too; if the rank
dropped further, all the cofactors would vanish and the derivative
relations would say nothing about the point.

A derivative of $f$ can have degree
$m-1$ in $x$. The expression on the right of~\eqref{eq:keyidentity}
is bilinear in two vectors of length $m$. The scalar $\lambda$ is
unknown, but it disappears from a homogeneous linear relation among
the derivatives. For example,
\[
  \partial_i f(a)-c\,\partial_j f(a)=0
  \quad\Longleftrightarrow\quad
  u^{\mathsf T}(A_i-cA_j)v=0.
\]

The points to count should therefore lie on $f=0$, be smooth points
of that hypersurface, and satisfy linear relations among the
derivatives.

\subsection{The number of conditions}

The equation $f=0$ leaves $n-1$ dimensions. To get isolated points,
I need $n-1$ more conditions. For a homogeneous polynomial, derivative
relations alone also respect scaling. Hence, if a nonzero point satisfies
them, then every nonzero scalar multiple does as well. I remove that
freedom by fixing one coordinate, say $x_{n-1}=1$.

That leaves room for $n-2$ derivative relations. For the power sum, a natural choice is to relate the first $n-2$
derivatives to the $(n-1)$st, because $\partial_i \F = n x_i^{n-1}$
involves only $x_i$. Once $x_{n-1}=1$, each relation reduces to
$x_i^{n-1}=c$ for a constant $c$.

Each of the first $n-2$ coordinates then has $n-1$ choices. The last
coordinate is determined by $\F=0$, which with $x_{n-1}=1$ reads
\[
  x_n^n=-\Bigl(1+\sum_{i=1}^{n-2}x_i^n\Bigr).
\]
This predicts the count
\[
  n(n-1)^{n-2}.
\]

The last equation might have a repeated root. An equation $t^n=R$ has $n$ distinct solutions
when $R\ne0$, but when $R=0$ its only solution is $t=0$, repeated $n$
times. Nothing so far prevents the right side above from vanishing.
With the plain relation $\partial_i\F=\partial_{n-1}\F$, each $x_i$ is
an $(n-1)$st root of unity, so $x_i^n=x_i$, and the right side becomes
$-1$ minus a sum of roots of unity, which can be zero. The smallest
case is $n=3$: taking $x_1=-1$ gives $x_3^3=0$
(Example~\ref{ex:cubic}).

For the exact lower bound, one can work with a weaker count of
distinct points, as in~\cite{KV26}. For the border bound, it is useful
to make every selected solution simple from the start. I do this by
choosing the first $n-2$ coordinates to have absolute value $1/2$.
Then
\[
  \Bigl|\sum_{i=1}^{n-2}x_i^n\Bigr|\le\frac{n-2}{2^n}<1,
\]
so the sum is too small to cancel the $1$, the right side is never
zero, and all $n$ values of $x_n$ are distinct.

\subsection{The quadratic count}

The kernel equations involve three groups of variables:
\[
  x\in\C^n,\qquad u\in\C^m,\qquad v\in\C^m.
\]
Each equation is linear in each group it uses. For example,
$u^{\mathsf T}A(x)=0$ is linear in $x$ and linear in $u$;
a derivative relation becomes an equation involving $u$ and $v$
but no $x$.

If I ignore these groups and remember only that the equations are
quadratic, ordinary B\'ezout gives an upper bound exponential in $m$.
Comparing it with roughly $n^n$ points yields a lower bound of order
$n\log n$.

The grouped B\'ezout count assigns each equation to one of the variable
groups it uses. Each group must receive as many equations as it has variables.
Once $x_{n-1}=1$ is fixed, there are only $n-1$ remaining $x$
variables. So exactly $n-1$ of the kernel equations must be assigned
to $x$; the others must account for $u$ and $v$.
There are about $2m$ kernel equations to choose from. This gives a
binomial coefficient, with the estimate
\[
  \binom{2m}{n-1}
  \le \left(\frac{2em}{n-1}\right)^{n-1}.
\]
The remaining choices, from the $n-2$ derivative relations, contribute
at most $2^{n-2}$. Section~\ref{sec:bezout} makes this count precise.
Together these give an upper bound of the form
\[
  \left(\frac{C m}{n}\right)^{n-1}.
\]
The power sum supplies at least $(n-1)^{n-1}$ points.
Taking $(n-1)$st roots leaves a quantity of order $n$ on that side,
and of order $m/n$ on the determinant side. Thus $m/n$ must be at
least a constant times $n$, which in turn forces $m$ to be at least a
constant times $n^2$.

\section{The power-sum system}
\label{sec:points}

Fix $n\ge3$ and put
\[
  \alpha=2^{-(n-1)}.
\]
For any polynomial $f$ in $n$ variables, define
\[
  P_i(f)=\partial_i f-\alpha\,\partial_{n-1}f,
  \qquad 1\le i\le n-2.
\]
I study the square system of $n$ equations
\begin{equation}
  f(x)=0,\qquad x_{n-1}=1,\qquad
  P_i(f)(x)=0\quad(1\le i\le n-2).
  \label{eq:slice}
\end{equation}
A solution is called \emph{simple} if the Jacobian matrix of these
$n$ equations with respect to the $n$ variables is invertible there.
This means that the equations meet independently at the point.
In particular, the solution is isolated.

\begin{lemma}
\label{lem:powerpoints}
For $f=\F$, the system~\eqref{eq:slice} has exactly
$N=n(n-1)^{n-2}$ solutions. Every solution is simple, and every
solution is a smooth zero of $\F$.
\end{lemma}

\begin{proof}
On the hyperplane $x_{n-1}=1$, the derivative equations become
\[
  x_i^{n-1}=2^{-(n-1)}
  \qquad(1\le i\le n-2).
\]
Therefore
$x_i=\zeta_i/2$, where $\zeta_i$ is an $(n-1)$st root of unity.
There are $(n-1)^{n-2}$ independent choices for these coordinates.
For each choice, the remaining equation is
\begin{equation}
  x_n^n=-\left(1+\sum_{i=1}^{n-2}x_i^n\right).
  \label{eq:lastcoordinate}
\end{equation}
Its right side is nonzero, since
\[
  \left|\sum_{i=1}^{n-2}x_i^n\right|
  \le \frac{n-2}{2^n}<1.
\]
Consequently~\eqref{eq:lastcoordinate} has exactly $n$ distinct
solutions, all nonzero. This proves the required count.

To check simplicity, first use $x_{n-1}=1$ to eliminate that
coordinate. The remaining variables are
$x_1,\ldots,x_{n-2},x_n$. The derivative equation indexed by $i$
has derivative
\[
  n(n-1)x_i^{n-2}
\]
in its own variable and zero in the other remaining variables.
These numbers are nonzero. The derivative of $\F$ with respect
to the final variable is $n x_n^{n-1}$, which is also nonzero.
The resulting Jacobian is triangular after ordering the equations
appropriately, with nonzero diagonal. The full Jacobian, including
$x_{n-1}-1$, is therefore invertible.

Finally, $\partial_{n-1}\F=n$ at every solution, so each point is
a smooth zero of $\F$.
\end{proof}

The number $N=n(n-1)^{n-2}$ is also the degree of the dual of the
Fermat hypersurface $\{\F=0\}\subset\mathbb P^{n-1}$, the variety
that records its tangent hyperplanes.
The original proof~\cite{She26} found the count through this geometric
interpretation; the explicit system above realizes it with simple solutions.

\begin{example}
\label{ex:cubic}
For $n=3$, write $\F=x^3+y^3+z^3$. The system is
\[
  x^3+y^3+z^3=0,\qquad y=1,\qquad
  3x^2-\frac14(3y^2)=0.
\]
Thus $x=1/2$ or $x=-1/2$. In the first case $z^3=-9/8$;
in the second $z^3=-7/8$. Each case gives three distinct values
of $z$, for a total of six simple solutions.

If I had instead required $\partial_x\F=\partial_y\F$,
I would obtain $x=\pm1$. At $x=-1$, the final equation would
be $z^3=0$. Three roots would coincide.
\end{example}

\section{Kernel equations}
\label{sec:kernels}

Define $f$ as the determinant of an $m\times m$ affine-linear matrix $A$,
as in~\eqref{eq:representation}. If I were to write
system~\eqref{eq:slice} directly for $f$, I would have to use $f$ itself,
of degree $m$, and its derivatives, of degree $m-1$. That runs into a
wall, since the B\'ezout bound for equations of that degree is too large
for my purpose. Therefore, consider the left and right kernel vectors $u$
and $v$ of $A(x)$. If I add them as extra variables, I replace
system~\eqref{eq:slice} with equations that have a degree of at most $2$,
since each equation is linear in each group of variables it involves:
$u^{\mathsf T}A(x)=0$ is linear in $x$ and in $u$, $A(x)v=0$ is linear in
$x$ and in $v$, and the derivative relations become
$u^{\mathsf T}(A_i-\alpha A_{n-1})v=0$, which is linear in $u$ and in $v$.

\subsection{Kernels and the adjugate}

\begin{lemma}
\label{lem:adjugate}
Suppose $\rank A(a)=m-1$.
For any nonzero left and right kernel vectors $u,v$ at $a$,
there is a scalar $\lambda\ne0$ such that
\[
  \adj A(a)=\lambda vu^{\mathsf T},
  \qquad
  \partial_i f(a)=\lambda u^{\mathsf T}A_i v
  \quad(1\le i\le n).
\]
In particular, these identities apply whenever $f(a)=0$ and
$\nabla f(a)\ne0$, since those conditions force $\rank A(a)=m-1$.
\end{lemma}

\begin{proof}
At a point of rank $m-1$, the identities
\[
  A(a)\adj A(a)=\adj A(a)A(a)=0
\]
show that every column of the adjugate lies in the right kernel
and every row lies in the left kernel. Both kernels are
one-dimensional, so the adjugate is $\lambda vu^{\mathsf T}$.
Some $(m-1)\times(m-1)$ minor is nonzero, so $\lambda\ne0$.
Differentiation of a determinant gives
\begin{equation}
  \partial_i f(x)=\tr\bigl(\adj A(x)\,A_i\bigr).
  \label{eq:jacobi}
\end{equation}
Substituting the factorization of the adjugate gives the
claimed derivative identities.

For completeness, formula~\eqref{eq:jacobi} follows by
differentiating the determinant with respect to each matrix entry.
The derivative with respect to entry $(j,k)$ is its cofactor,
and the derivative of that entry with respect to $x_i$ is
$(A_i)_{jk}$. Summing the products of these two quantities is
exactly the trace in~\eqref{eq:jacobi}.

Finally, suppose $f(a)=0$ and $\nabla f(a)\ne0$. The rank
is at most $m-1$. If it were at most $m-2$, all
$(m-1)\times(m-1)$ minors would vanish, so $\adj A(a)=0$.
Formula~\eqref{eq:jacobi} would then force every derivative
of $f$ to vanish at $a$. Hence the rank is exactly $m-1$.
\end{proof}

\subsection{Normalizing the kernel vectors}

Define $S$ as a finite set of isolated solutions of~\eqref{eq:slice},
with $\nabla f(a)\neq 0$ at each point $a\in S$. At each such point, the
left and right kernels of $A(a)$ are lines, so $u$ and $v$ are determined
only up to a nonzero multiple. To count with B\'ezout, each point must
give a single isolated solution, so I fix the scale of $u$ and $v$.

Choose vectors $q,p\in\C^m$ such that
\[
  u^{\mathsf T}q\ne0,\qquad p^{\mathsf T}v\ne0
\]
on every left and right kernel line belonging to a point of $S$.
Such choices exist: each line excludes just one hyperplane of
possible choices, and finitely many proper hyperplanes cannot
cover $\C^m$. I then normalize the kernel vectors by
\begin{equation}
  u^{\mathsf T}q=1,\qquad p^{\mathsf T}v=1.
  \label{eq:normalizations}
\end{equation}
Each point of $S$ now has exactly one normalized pair $(u,v)$.

\subsection{A redundant equation}

The full left and right kernel conditions have a total of $2m$ equations,
but contain a redundancy. The scalar $u^{\mathsf T}A(x)v$ can be obtained
from either block. In order to apply the square-system B\'ezout bound, I
must remove the redundancy carefully.

Choose a linear map
\[
  \Lambda:\C^m\longrightarrow\C^{m-1},
  \qquad \ker\Lambda=\C q.
\]
Replace the $m$ equations $A(x)v=0$ by the $m-1$ equations
$\Lambda A(x)v=0$. Under the other equations, nothing is lost.
The reduced equations imply $A(x)v=cq$ for some scalar
$c$. Multiplying by $u^{\mathsf T}$ gives
\[
  0=u^{\mathsf T}A(x)v=c\,u^{\mathsf T}q=c.
\]
Thus $A(x)v=0$ after all.

\subsection{The incidence system}

Write $B_i=A_i-\alpha A_{n-1}$ for $1\le i\le n-2$.
Consider the following equations in $(x,u,v)$:
\begin{equation}
\begin{aligned}
  u^{\mathsf T}A(x)&=0,\\
  \Lambda A(x)v&=0,\\
  u^{\mathsf T}B_i v&=0 &&(1\le i\le n-2),\\
  x_{n-1}&=1,\\
  u^{\mathsf T}q&=1,\qquad p^{\mathsf T}v=1.
\end{aligned}
\label{eq:incidence}
\end{equation}
There are $n+2m$ variables and
\[
  m+(m-1)+(n-2)+3=n+2m
\]
scalar equations. Every point of $S$ gives a solution:
the kernel equations and normalizations hold by construction,
and Lemma~\ref{lem:adjugate} converts each derivative relation
into the corresponding equation involving $B_i$.

The system may have other solutions. Some may lie on curves or larger
sets, especially where the rank of $A(x)$ drops further. I single out
solutions from $S$, and only count these in the larger system.

\begin{lemma}
\label{lem:isolatedlift}
Every point $a\in S$ gives a distinct isolated solution of
\eqref{eq:incidence}.
\end{lemma}

\begin{proof}
Let $(a,u_a,v_a)$ be the normalized solution above $a$.
Some $(m-1)\times(m-1)$ minor of $A(a)$ is nonzero.
It remains nonzero in a neighborhood of $a$, so throughout that
neighborhood the rank is at least $m-1$.

At any solution of~\eqref{eq:incidence}, the normalization makes
$u$ nonzero. The left kernel equation therefore implies
$\det A(x)=0$. For solutions near $(a,u_a,v_a)$, the rank
is consequently exactly $m-1$. The reduced right equations
imply the full right kernel condition, as shown above.
Lemma~\ref{lem:adjugate} now applies, so
\[
  P_i(f)(x)=\lambda u^{\mathsf T}B_i v=0.
\]
Together with $f(x)=0$ and $x_{n-1}=1$, these are precisely
the equations of~\eqref{eq:slice}.

Since $a$ is an isolated solution of that system, the only
possible $x$ sufficiently near $a$ is $a$ itself. The normalized
kernel vectors at $a$ are unique. Hence the only solution
sufficiently near $(a,u_a,v_a)$ is that triple.
Different points of $S$ have different first coordinates, so
their lifts are distinct.
\end{proof}

\section{The multihomogeneous count}
\label{sec:bezout}

\subsection{The multihomogeneous B\'ezout theorem}

Suppose the variables are divided into groups of sizes
$r_1,\ldots,r_k$, with $r=r_1+\cdots+r_k$ in total.
For each of $r$ equations, record its degree in each group
separately. If those degrees are at most
$(d_{j1},\ldots,d_{jk})$ for equation $j$, associate to it
the formal linear expression
\[
  d_{j1}Y_1+\cdots+d_{jk}Y_k.
\]
The symbols $Y_i$ are only used to calculate the bound.

\begin{theorem}[Multihomogeneous B\'ezout bound]
\label{thm:bezout}
The number of isolated complex solutions of the system is at most
\begin{equation}
  \coeff{Y_1^{r_1}\cdots Y_k^{r_k}}
  {\prod_{j=1}^{r}(d_{j1}Y_1+\cdots+d_{jk}Y_k)}.
  \label{eq:generalbezout}
\end{equation}
The bound remains valid when the system also has
positive-dimensional sets of solutions.
\end{theorem}

This is a standard form of the multihomogeneous B\'ezout theorem;
see Morgan and Sommese~\cite{MS87} and
Sommese and Wampler~\cite[Chapter~8]{SW05}.

Each group of variables has as many dimensions as variables, and the
equations have to cut all of them down to points. Expanding the product
implies letting each equation spend its condition on one of its groups.
The coefficient keeps only those expansions where every group receives
exactly as many conditions as variables.

For example, consider two equations in variables $s,t$, each of
degree at most one in $s$ and at most one in $t$. The ordinary
degree bound is four, since both equations can have total degree
two. The grouped bound is
\[
  [ST](S+T)^2=2.
\]
Indeed, solving one equation for $t$ and substituting into the other
gives an equation of degree at most two in $s$. The coefficient
rule is the higher-dimensional version of this improvement.

\subsection{The B\'ezout coefficient}

Apply Theorem~\ref{thm:bezout} to the groups $x,u,v$, of
sizes $n,m,m$. Use formal symbols $X,U,V$ for their degrees.
The complete degree calculation is:
\begin{center}
\renewcommand{\arraystretch}{1.18}
\begin{tabular}{@{}lcc@{}}
\toprule
Equations & Number & Degree expression\\
\midrule
$u^{\mathsf T}A(x)=0$ & $m$ & $X+U$\\
$\Lambda A(x)v=0$ & $m-1$ & $X+V$\\
$u^{\mathsf T}B_i v=0$ & $n-2$ & $U+V$\\
$x_{n-1}-1=0$ & $1$ & $X$\\
$u^{\mathsf T}q-1=0$ & $1$ & $U$\\
$p^{\mathsf T}v-1=0$ & $1$ & $V$\\
\bottomrule
\end{tabular}
\end{center}
It follows that
\begin{align}
  |S|
  &\le [X^nU^mV^m]\,
  XUV(X+U)^m(X+V)^{m-1}(U+V)^{n-2}\notag\\
  &= [X^{n-1}U^{m-1}V^{m-1}]\,
  (X+U)^m(X+V)^{m-1}(U+V)^{n-2}.
  \label{eq:Bcoefficient}
\end{align}
Call the last coefficient $B(m,n)$.

To expand it, suppose I choose $X^kU^{m-k}$ from the first
factor. The power of $X$ then forces the choice of
$X^{n-1-k}V^{m-n+k}$ from the second factor. The required
power of $U$ forces the choice $U^{k-1}V^{n-1-k}$ from the
third factor. Therefore
\begin{equation}
  B(m,n)=
  \sum_{k=1}^{n-1}
  \binom{m}{k}\binom{m-1}{n-1-k}\binom{n-2}{k-1}.
  \label{eq:Bsum}
\end{equation}
As usual, a binomial coefficient is zero when its lower index
is negative or exceeds its upper index.

\begin{proposition}
\label{prop:detcount}
Let $f=\det A(x)$ for an $m\times m$ affine-linear matrix.
If $S$ is a finite set of isolated solutions of~\eqref{eq:slice}
at which $\nabla f\ne0$, then $|S|\le B(m,n)$.
\end{proposition}

\subsection{Comparison of the counts}

The exact expression~\eqref{eq:Bsum} is not needed to read off
the order of growth. Bound
$\binom{n-2}{k-1}\le2^{n-2}$ and use Vandermonde's identity:
\[
  B(m,n)\le
  2^{n-2}\sum_k\binom{m}{k}\binom{m-1}{n-1-k}
  =2^{n-2}\binom{2m-1}{n-1}.
\]
For $m\ge n$, the elementary estimate
$\binom{a}{b}\le(ea/b)^b$ gives
\begin{equation}
  B(m,n)\le
  2^{n-2}\left(\frac{e(2m-1)}{n-1}\right)^{n-1}
  \le \left(\frac{4em}{n-1}\right)^{n-1}.
  \label{eq:Bestimate}
\end{equation}

For an exact representation of $\F$, $m\ge n$ by
degree. Lemma~\ref{lem:powerpoints} and
Proposition~\ref{prop:detcount} now give
\[
  n(n-1)^{n-2}
  \le B(m,n)
  \le \left(\frac{4em}{n-1}\right)^{n-1}.
\]
Since $n(n-1)^{n-2}\ge(n-1)^{n-1}$, taking $(n-1)$st roots
yields
\[
  n-1\le\frac{4em}{n-1},
  \qquad\text{and hence}\qquad
  m\ge\frac{(n-1)^2}{4e}.
\]
This proves the ordinary exact lower bound.

\section{Border approximation}
\label{sec:border}

An exact lower bound need not be a border lower bound. A polynomial could fail to have a small representation and still be a limit of polynomials that do have one.

\subsection{Persistence of simple solutions}

For a polynomial $f$ define
\[
  \Phi_f(x)=
  \bigl(f(x),\,x_{n-1}-1,\,P_1(f)(x),\ldots,P_{n-2}(f)(x)\bigr).
\]
Fix a degree bound $d\ge n$ and regard polynomials of degree at most
$d$ as vectors of coefficients. The map $(f,x)\mapsto\Phi_f(x)$
is polynomial in those coefficients and in $x$.

\begin{lemma}
\label{lem:persistence}
There is a neighborhood $\mathcal U$ of $\F$ in the coefficient
space of polynomials of degree at most $d$ such that, for every
$f\in\mathcal U$, the system $\Phi_f(x)=0$ has at least
$N=n(n-1)^{n-2}$ distinct simple solutions. Each is a smooth
zero of $f$.
\end{lemma}

\begin{proof}
Put a small ball around each solution $a_1,\ldots,a_N$ from
Lemma~\ref{lem:powerpoints}, with the balls chosen disjoint.
At each $(\F,a_j)$, the Jacobian of $\Phi_f$ with respect to $x$
is invertible. The implicit function theorem therefore gives a
solution $a_j(f)$ in that ball whenever the coefficients of $f$
are sufficiently close to those of $\F$.

There are only finitely many balls, so I can choose one coefficient
neighborhood $\mathcal U$ that works for all of them. The balls keep
the solutions distinct. By continuity, shrinking $\mathcal U$ if
necessary keeps every Jacobian invertible, so the solutions remain simple.

The first row of that Jacobian is $\nabla f(a_j(f))$.
An invertible matrix cannot have a zero row. Hence
$\nabla f(a_j(f))\ne0$ for every $j$.
\end{proof}

\subsection{Proof of the border bound}

Suppose $\bdc(\F)\le m$. Every size-$m$ determinant has
degree at most $m$, and this degree condition is preserved
under coefficientwise limits. Therefore $m\ge n$.
Choose a sequence
\[
  f_j(x)=\det\left(A_{j,0}+\sum_{i=1}^n x_iA_{j,i}\right)
  \longrightarrow \F
\]
in the coefficient space of polynomials of degree at most $m$.

Use Lemma~\ref{lem:persistence} with $d=m$.
For all sufficiently large $j$, the polynomial $f_j$ belongs
to $\mathcal U$ and has at least $N$ simple solutions of
\eqref{eq:slice}. At those points its gradient is nonzero.
Proposition~\ref{prop:detcount} therefore applies to the
size-$m$ representation of $f_j$ and gives
\[
  n(n-1)^{n-2}\le B(m,n).
\]
The calculation in~\eqref{eq:Bestimate} is unchanged, and yields
\[
  m\ge\frac{(n-1)^2}{4e}.
\]
This proves the ordinary border lower bound.

Only the polynomials $f_j$ had to converge. For each fixed $j$, I first
find the $N$ nearby points, and then choose the kernel vectors and
normalizations. Note here that the upper bound $B(m,n)$ only depends on
the dimensions and degrees of the equations, so it is the same for every
member of the sequence.

\section{Symmetric representations}
\label{sec:symmetric}

For a symmetric matrix, its left and right kernels are the same, so I
only need one of them. This lowers the count of auxiliary variables,
and is the source of the improved constant.

Suppose $f=\det A(x)$ with each $A_i$ symmetric. At a smooth
zero $a$ of $f$, choose a nonzero vector $u$ in the kernel.
The adjugate is a nonzero symmetric rank-one matrix with image
$\C u$, so
\[
  \adj A(a)=\lambda uu^{\mathsf T}
  \quad\text{for some }\lambda\ne0.
\]
Hence
\[
  \partial_i f(a)=\lambda u^{\mathsf T}A_i u.
\]

Let $S$ again be a finite set of smooth isolated solutions
of~\eqref{eq:slice}. Choose $q$ nonorthogonal to the finitely
many kernel lines and normalize by $q^{\mathsf T}u=1$.
The equations are now
\begin{equation}
\begin{aligned}
  A(x)u&=0,\\
  u^{\mathsf T}(A_i-\alpha A_{n-1})u&=0
    &&(1\le i\le n-2),\\
  x_{n-1}&=1,\qquad q^{\mathsf T}u=1.
\end{aligned}
\label{eq:symincidence}
\end{equation}
There are $n+m$ variables and $m+(n-2)+2=n+m$ equations.
No kernel equation needs to be removed.

Isolation follows from this as in Lemma~\ref{lem:isolatedlift}.

There are now only two variable groups, of sizes $n$ and $m$.
The kernel equations have degree expression $X+U$,
the $n-2$ derivative equations have degree expression $2U$,
and the normalizations contribute $X$ and $U$. Thus
\begin{align*}
  |S|
  &\le [X^nU^m]\,XU(X+U)^m(2U)^{n-2}\\
  &=2^{n-2}\binom{m}{n-1}.
\end{align*}
The choice of $n-1$ factors of $X$ from $(X+U)^m$
accounts for the binomial coefficient.

For $f=\F$, or for a sufficiently close symmetric determinant
by Lemma~\ref{lem:persistence}, this gives
\[
  n(n-1)^{n-2}\le 2^{n-2}\binom{m}{n-1}
  \le\left(\frac{2em}{n-1}\right)^{n-1}.
\]
Taking $(n-1)$st roots as before proves
\[
  \bsdc(\F)\ge\frac{(n-1)^2}{2e}.
\]
The exact symmetric lower bound follows as well.

\section{Upper bounds}
\label{sec:upper}

Put $r=n-1$, and let $J$ be the $r\times r$ matrix with
ones immediately above the diagonal and zeros elsewhere.
For one variable $t$, set
\[
  T(t)=I_r-tJ.
\]
This is an upper triangular matrix with determinant one.
Since $J^r=0$,
\[
  T(t)^{-1}=I_r+tJ+\cdots+t^{r-1}J^{r-1}.
\]
In particular, if $e_1,e_r$ are the first and last coordinate
vectors,
\[
  e_1^{\mathsf T}T(t)^{-1}e_r=t^{r-1}=t^{n-2}.
\]

\subsection{An ordinary representation}

Let $T_i=T(x_i)$, and form
\[
  D=\operatorname{diag}(T_1,\ldots,T_n),\qquad
  b^{\mathsf T}=(x_1e_1^{\mathsf T},\ldots,x_ne_1^{\mathsf T}),
  \qquad
  c=\begin{pmatrix}x_1e_r\\ \vdots\\ x_ne_r\end{pmatrix}.
\]
The matrix
\[
  M(x)=
  \begin{pmatrix}
    0&b^{\mathsf T}\\
    -c&D
  \end{pmatrix}
\]
has affine-linear entries and size $1+n(n-1)$. The block
determinant formula gives
\[
  \det M=\det D\,(b^{\mathsf T}D^{-1}c)
  =\sum_{i=1}^n x_i^2
      e_1^{\mathsf T}T_i^{-1}e_r
  =\sum_{i=1}^n x_i^n.
\]
Thus $\dc(\F)\le1+n(n-1)$.

I sum the contributions of the $n$ blocks by adding the extra row and
column. Each block's inverse contributes $x_i^{n-2}$. The two connecting
entries from $b$ and $c$ add the remaining factor $x_i^2$.

\subsection{A symmetric representation}

The same construction can be made symmetric by pairing
each block with its transpose. Define
\[
  C_i=
  \begin{pmatrix}
    0&T_i^{\mathsf T}\\ T_i&0
  \end{pmatrix},
  \qquad
  w_i=
  \begin{pmatrix}
    -\tfrac12 x_i e_1\\ x_i e_r
  \end{pmatrix}.
\]
Then $C_i$ is symmetric and
\[
  C_i^{-1}=
  \begin{pmatrix}
    0&T_i^{-1}\\ T_i^{-\mathsf T}&0
  \end{pmatrix}.
\]
The block order can be checked directly:
\[
  C_iC_i^{-1}=
  \begin{pmatrix}
    T_i^{\mathsf T}T_i^{-\mathsf T}&0\\
    0&T_iT_i^{-1}
  \end{pmatrix}
  =I_{2r}.
\]
The two cross terms in the quadratic expression are equal, so
\[
  w_i^{\mathsf T}C_i^{-1}w_i
  =2\left(-\tfrac12 x_i e_1\right)^{\mathsf T}
       T_i^{-1}(x_i e_r)
  =-x_i^2 e_1^{\mathsf T}T_i^{-1}e_r
  =-x_i^n.
\]
Let $C=\operatorname{diag}(C_1,\ldots,C_n)$ and stack
the $w_i$ into a column $w$. Since
$\det C_i=(-1)^r$ and $nr=n(n-1)$ is even, $\det C=1$.
Therefore the symmetric affine-linear matrix
\[
  M_{\mathrm{sym}}(x)=
  \begin{pmatrix}0&w^{\mathsf T}\\ w&C\end{pmatrix}
\]
satisfies
\[
  \det M_{\mathrm{sym}}
  =-\det C\, w^{\mathsf T}C^{-1}w
  =\F.
\]
Its size is $1+2n(n-1)$. These two constructions, together
with the lower bounds, complete Theorem~\ref{thm:main}.

\subsection{Power sums of any degree}
\label{sec:anydegree}

Nothing in the argument used that the degree equals the number of
variables. For $N\ge3$ and $d\ge2$, let
\[
  F_{N,d}=x_1^d+\cdots+x_N^d .
\]

\begin{corollary}
\label{cor:anydegree}
For all $N\ge3$ and $d\ge2$, over $\C$,
\[
  \frac{(N-1)(d-1)}{4e}\le\bdc(F_{N,d})\le\dc(F_{N,d})\le 1+N(d-1),
\]
\[
  \frac{(N-1)(d-1)}{2e}\le\bsdc(F_{N,d})\le\sdc(F_{N,d})\le 1+2N(d-1).
\]
In particular, all four complexities of $F_{N,d}$ are $\Theta(Nd)$.
\end{corollary}

The lower bound in Corollary~\ref{cor:anydegree} also follows from the
general hypersurface bound in~\cite{SheCeiling26}; the argument here
gives a short direct proof.

\begin{proof}
Choose $\rho>0$ with $(N-2)\rho^d<1$ and use the system~\eqref{eq:slice}
with $N$ in place of $n$ and $\alpha=\rho^{d-1}$. On $x_{N-1}=1$ the
derivative relations become $x_i^{d-1}=\alpha$, so $x_i=\rho\zeta_i$
with $\zeta_i^{d-1}=1$, giving $(d-1)^{N-2}$ choices. The last
coordinate satisfies $x_N^d=-\bigl(1+\sum_{i\le N-2}x_i^d\bigr)$, and
the right side is nonzero because
$\bigl|\sum_{i\le N-2}x_i^d\bigr|\le(N-2)\rho^d<1$. The proof of
Lemma~\ref{lem:powerpoints} then gives exactly $d(d-1)^{N-2}$
solutions, all simple and all smooth zeros. The proof of
Lemma~\ref{lem:persistence} uses only simplicity, and
Proposition~\ref{prop:detcount} holds with $N$ in place of $n$, so
\[
  (d-1)^{N-1}\le d(d-1)^{N-2}\le B(m,N)
  \le\left(\frac{4em}{N-1}\right)^{N-1}.
\]
Taking $(N-1)$st roots gives $m\ge(N-1)(d-1)/(4e)$. The symmetric
bound follows in the same way from
$d(d-1)^{N-2}\le2^{N-2}\binom{m}{N-1}\le\bigl(2em/(N-1)\bigr)^{N-1}$.

For the upper bounds, repeat Section~\ref{sec:upper} with blocks of
size $r=d-1$, so that $e_1^{\mathsf T}T(t)^{-1}e_r=t^{d-2}$. This gives
an ordinary representation of size $1+N(d-1)$ and a symmetric one of
size $1+2N(d-1)$. In the symmetric case, replace $-\tfrac12$ by
$+\tfrac12$ in $w_i$ when $N(d-1)$ is odd, since then $\det C=-1$.
\end{proof}

For $d=N=n$ this is Theorem~\ref{thm:main}. The bound is superlinear in
$\max(N,d)$ whenever both grow, which is the regime singled out as open
in~\cite{KV21}.

\section{Discussion}
\label{sec:perspective}

Note that a high degree alone does not imply many usable points. For
example, $x_1\cdots x_n$ is the determinant of an $n\times n$ diagonal
matrix. Taking a smooth zero of the polynomial, exactly one coordinate is
zero, implying that the gradient has only one nonzero entry. So its
tangent directions barely vary. At singular points, many derivative
equations automatically vanish, and this results in larger sets rather
than isolated points.

Similarly, the permanent is highly singular, and the same issue arises.
The permanent would hence need a fresh count of usable points on its
smooth part. However, multiplying the degrees of its derivative equations
only yields an upper bound on the count, without saying anything about
the existence of usable points.

In~\cite[Theorem~7.7 and Corollary~7.11]{SheCeiling26}, I study
how far this kind of counting argument can go.
For the polar method and the class of refinements covered there,
the lower bounds it can prove are at most $O(dN)$ for a homogeneous
polynomial of degree $d$ in $N$ variables.
The $n\times n$ permanent has degree $n$ and $n^2$ variables, so
the ceiling is $O(n^3)$. A supercubic lower bound would need
something beyond that class of arguments.

For $\F$ too, this argument gives $\Theta(n^2)$, not the optimal
constant. Whereas B\'ezout counts all isolated solutions of the kernel
system, the selected points may be only some of them. Additionally, the
binomial estimates loosen the bound further. Hence, improving the
constant would require an exact understanding of the kernel system, or a
new count.


\begin{thebibliography}{99}

\bibitem{BS83}
Walter Baur and Volker Strassen.
\newblock The complexity of partial derivatives.
\newblock \emph{Theoretical Computer Science}, 22:317--330, 1983.

\bibitem{BI25}
Markus Bl\"aser and Christian Ikenmeyer.
\newblock Introduction to geometric complexity theory.
\newblock \emph{Theory of Computing Library, Graduate Surveys},
10:1--166, 2025.
\newblock \href{https://doi.org/10.4086/toc.gs.2025.010}
{doi:10.4086/toc.gs.2025.010}.

\bibitem{KV21}
Mrinal Kumar and Ben Lee Volk.
\newblock A lower bound on determinantal complexity.
\newblock In \emph{36th Computational Complexity Conference},
LIPIcs 200, 4:1--4:12, 2021.
\newblock \href{https://doi.org/10.4230/LIPIcs.CCC.2021.4}
{doi:10.4230/LIPIcs.CCC.2021.4}.

\bibitem{KV26}
Mrinal Kumar and Ben Lee Volk.
\newblock A quadratic lower bound on determinantal complexity.
\newblock \emph{Electronic Colloquium on Computational Complexity},
Report TR26-218, 2026.
\newblock \href{https://arxiv.org/abs/2609.34462}
{arXiv:2609.34462}.

\bibitem{MR04}
Thierry Mignon and Nicolas Ressayre.
\newblock A quadratic bound for the determinant and permanent problem.
\newblock \emph{International Mathematics Research Notices},
2004(79):4241--4253, 2004.

\bibitem{MS87}
Alexander Morgan and Andrew Sommese.
\newblock A homotopy for solving general polynomial systems that
respects $m$-homogeneous structures.
\newblock \emph{Applied Mathematics and Computation},
24(2):101--113, 1987.

\bibitem{She26}
Karthik Sheshadri.
\newblock A fully Lean-verified quadratic lower bound on the border
determinantal complexity of $x_1^n+\cdots+x_n^n$ over $\C$,
including the symmetric and exact models.
\newblock \href{https://arxiv.org/abs/2606.13628v3}
{arXiv:2606.13628v3}, 2026.

\bibitem{SW05}
Andrew J. Sommese and Charles W. Wampler.
\newblock \emph{The Numerical Solution of Systems of Polynomials
Arising in Engineering and Science}.
\newblock World Scientific, 2005.

\bibitem{Val79}
Leslie G. Valiant.
\newblock Completeness classes in algebra.
\newblock In \emph{Proceedings of the 11th Annual ACM Symposium
on Theory of Computing}, pages 249--261, 1979.

\bibitem{SheCeiling26}
Karthik Sheshadri.
\newblock Conormal lower bounds for determinantal complexity:
Fermat slices, tight bounds, and a ceiling.
\newblock \href{https://arxiv.org/abs/2606.15970v2}
{arXiv:2606.15970v2}, 2026.

\end{thebibliography}
\end{document}